\documentclass[a4paper,twoside]{article}
\usepackage{amsmath,amssymb,mathrsfs,bm}
\newcommand{\realnumbers}{\mathbb{R}}
\newcommand{\integers}{\mathbb{Z}}
\newcommand{\pa}{\partial}
\bmdefine{\bZ}{Z}

\usepackage{geometry}

\usepackage{fancyhdr}
\newcommand{\pfa}[2]{\frac{\pa #1}{\pa p_a}\wedge\frac{\pa #2}{\pa\varphi^a}}
\newcommand{\fpa}[2]{\frac{\pa #1}{\pa\varphi^a}\wedge\frac{\pa #2}{\pa p_a}}
\newcommand{\psa}[2]{\frac{\pa #1}{\pa\pi_\alpha}\wedge\frac{\pa #2}{\pa\psi^\alpha}}
\newcommand{\spa}[2]{\frac{\pa #1}{\pa\psi^\alpha}\wedge\frac{\pa #2}{\pa\pi_\alpha}}

\newcommand{\scrO}{\mathscr{O}}
\bmdefine{\bpa}{\partial}
\bmdefine{\bwedge}{\wedge}
\bmdefine{\blambda}{\lambda}
\bmdefine{\bomega}{\omega}
\bmdefine{\bd}{d}
\bmdefine{\bX}{X}
\bmdefine{\bY}{Y}
\newcommand{\Hom}{\mathop{\rm Hom}\nolimits}
\newcommand{\Der}{\mathop{\rm Der}\nolimits}
\newcommand{\la}{\langle}
\newcommand{\ra}{\rangle}
\usepackage{cite}

\usepackage{theorem}
\newtheorem{Thm}{Theorem}
\newtheorem{Def}[Thm]{Definition}
\newtheorem{Prop}[Thm]{Proposition}

\newenvironment{proof}{\noindent\textit{Proof.}}{\hfill$\square$}

\begin{document}

\title{
\vspace{-12ex}
\begin{flushright}
{\normalsize EMPG--17--03}
\end{flushright}
\vspace{9ex}
\textbf{Poisson Bracket and Symplectic Structure of Covariant Canonical Formalism of Fields}
} 

\author{\textbf{Yasuhito Kaminaga\footnote{
Permanent Address: 
Department of Mathematics, 
National Institute of Technology, Gunma College, 
Toriba, 
Maebashi, Gunma 371--8530, Japan
~
E-mail: kaminaga@nat.gunma-ct.ac.jp
}} 
\\ \normalsize 
School of Mathematics, 
The University of Edinburgh, \\ \normalsize 
James Clerk Maxwell Building, Peter Guthrie Tait Road,  
\\ \normalsize 
Edinburgh EH9 3FD, United Kingdom
}
\date{}
\maketitle

\begin{abstract}
The covariant canonical formalism 
is a covariant extension of the traditional canonical formalism of fields. 
In contrast to the traditional canonical theory, 
it has a remarkable feature that canonical equations 
of gauge theories or gravity are not only manifestly Lorentz covariant 
but also gauge covariant or diffeomorphism covariant. 
A mathematical peculiarity of the covariant canonical formalism 
is that its canonical coordinates are differential forms on a manifold. 
In the present paper, we find a natural Poisson bracket of 
this new canonical theory, and study symplectic structure behind it. 
The phase space of the theory is identified with a ringed space 
with the structure sheaf of the graded algebra of ``differentiable'' 
differential forms on the manifold. 
The Poisson and the symplectic structure we found can be even or odd, 
depending on the dimension of the manifold. 
Our Poisson structure is an example of physical application of Poisson structure 
defined on the graded algebra of differential forms. 
\end{abstract}

\begin{flushleft} 
{\it 
Keywords: Differential Form; Analytical Mechanics; Gauge Theory; Gravity; 
Graded Lie Algebra; Graded Manifold; Poisson Bracket; Symplectic Structure 
\\
PACS (2010): 
02.40.-k, 
03.50.-z, 
04.20.-q, 
04.20.Fy, 
11.10.Ef, 
11.15.-q  
}
\end{flushleft}
\newpage

\section{Introduction}

In the previous paper \cite{Kaminaga}, 
the present author developed a new analytical mechanics of fields 
which treats space and time on an equal footing \cite{Nakamura,Nester,Nakajima1,Nakajima2}. 
A similar theory has been studied by Th\'eophile De Donder \cite{DeDonder}, Hermann Weyl \cite{Weyl} 
and others \cite{Kastrup,Kanatchikov,CrnkovicWitten,ForgerRomero,ForgerSalles,Sharapov}, 
but there is a crucial difference between theirs and ours. 
That is, they adopt components of tensors (not tensors themselves) as canonical variables 
while we adopt differential forms themselves (not their components) as canonical variables. 
This difference causes drastic changes to the resultant theory at least formally.  
The De Donder-Weyl theory including multisymplectic formalism 
\cite{DeDonder,Weyl,Kastrup,Kanatchikov,CrnkovicWitten,ForgerRomero,ForgerSalles,Sharapov} 
is more similar to the traditional analytical mechanics of fields than ours. 
For example, in their theory one needs gauge fixing or Dirac bracket or something like that 
to obtain their canonical equations of gauge theories or those of gravity, 
as it is the case with the traditional canonical formalism. 
In our theory, however, one can obtain 
canonical equations of gauge theories or those of gravity 
without fixing a gauge nor introducing Dirac bracket 
nor any other artificial tricks. 
Our canonical equations are not only manifestly Lorentz covariant 
but also keep any kind of gauge freedom including diffeomorphism. 
Our new analytical mechanics of fields is impressive with its simplicity, 
straightforward nature and mathematical beauty. 
The purpose of the present paper is to study Poisson structure of our 
canonical formalism, and to clarify symplectic structure behind it. 

This paper is organised as follows. 
In Section \ref{sec:Covariant Canonical Formalism of Fields} 
we give a brief review of our covariant canonical formalism \cite{Kaminaga}. 
In Section \ref{sec:Finding Poisson Bracket}, 
we find a natural definition of our Poisson bracket on the basis of 
physical motivations and heuristic speculations. 
In Section \ref{sec:List of Poisson Bracket Formulae}, 
we collect important formulae of our Poisson bracket. 
The heuristic part ends here and 
purely mathematical part begins at Section \ref{sec:Phase Space}. 
In Sections \ref{sec:Phase Space} and \ref{sec:Lie Derivatives}, 
we develop mathematical tools needed for the study of our symplectic structure. 
In Section \ref{sec:Symplectic Structure}, 
we study symplectic structure of our covariant canonical formalism 
using tools prepared in Sections \ref{sec:Phase Space} and \ref{sec:Lie Derivatives},  
and show that it leads to our Poisson bracket found in the heuristic part.

\section{Covariant Canonical Formalism of Fields}
\label{sec:Covariant Canonical Formalism of Fields}

Let $M$ be a $C^\infty$-manifold of dimension $n$, and let 
$\Omega(M)=\bigoplus_{r=0}^{n}\Omega^{r}(M)$ be the 
$\integers_2$-graded commutative algebra of differential forms on $M$. 
For a map $\phi:\Omega^{p_1}(M)\times\cdots\times\Omega^{p_k}(M)\to\Omega^{r}(M)$, 
we put $\beta=\phi(\alpha^1,\cdots,\alpha^k)\in\Omega^r(M)$ with $\alpha^i\in\Omega^{p_i}(M)$. 
If there exists a differential form $\omega_i\in\Omega^{r-p_i}(M)$ 
such that $\beta$ behaves under variations $\delta\alpha^1,\cdots,\delta\alpha^k$ 
as $\delta\beta=\sum_i\delta\alpha^i\wedge\omega_i$, then we say 
the differential form $\beta$ is \textit{differentiable} 
with respect to another differential form $\alpha^i$, 
and define the derivative of $\beta$ by $\alpha^i$ to be 
$\pa \beta/\pa \alpha^i=\omega^i$. 
In what follows, we assume sum over repeated indices unless otherwise stated. 

Suppose physical fields are represented with 
even-forms $\varphi^a\in\Omega(M)$ and/or odd-forms $\psi^\alpha\in\Omega(M)$. 
We define a Lagrangian $n$-form $L\in\Omega(M)$ 
so that the covariant Lagrange equations 
\begin{align*}
  \frac{\pa L}{\pa\varphi^a}-d\frac{\pa L}{\pa d\varphi^a}=0, \quad 
  \frac{\pa L}{\pa\psi^\alpha}+d\frac{\pa L}{\pa d\psi^\alpha}=0 
\end{align*}
coincide with field equations. 
Here, we have implicitly assumed, of course, 
that $L$ is differentiable with respect to $\varphi^a$, $d\varphi^a$, $\psi^\alpha$ and $d\psi^\alpha$. 
Then, we define the conjugate momentum forms $p_a, \pi_\alpha\in\Omega(M)$ as 
\begin{align}
  p_a=\frac{\pa L}{\pa d\varphi^a}, \quad 
  \pi_\alpha=\frac{\pa L}{\pa d\psi^\alpha}, 
  \label{ph:def_of_momentum}
\end{align}
and the Hamiltonian $n$-form $H\in\Omega(M)$ as 
\begin{align*}
  H=d\varphi^a\wedge p_a +d\psi^\alpha\wedge\pi_\alpha -L . 
\end{align*}
Suppose that the system is non-singular, that is, 
(\ref{ph:def_of_momentum}) can be solved conversely and 
$d\varphi^a$ and $d\psi^\alpha$ are represented uniquely with 
canonical variables $\varphi^a$, $p_a$, $\psi^\alpha$ and $\pi_\alpha$. 
Then, $H$ is differentiable with respect to canonical variables, and 
we get our covariant canonical equations 
\begin{align}
  d\varphi^a = -(-1)^n\frac{\pa H}{\pa p_a}, \qquad
  dp_a = -\frac{\pa H}{\pa\varphi^a}, \qquad
  d\psi^\alpha = \frac{\pa H}{\pa \pi_\alpha}, \qquad
  d\pi_\alpha = \frac{\pa H}{\pa\psi^\alpha}.
  \label{ph:canonical_equations}
\end{align}
Surprisingly, the non-singular assumption is satisfied 
not only by gauge theories but also by gravity \cite{Kaminaga}. 
Indeed, our canonical equations (\ref{ph:canonical_equations}) 
are manifestly Lorentz covariant as well as gauge/diffeomorphism covariant. 

In the following sections, we will denote by $\Lambda^r(M)$ 
the set of $r$-forms, on $M$,  
which are sufficiently many times differentiable 
with respect to canonical variables $\varphi^a$, $p_a$, $\psi^\alpha$ and $\pi_\alpha$. 
$\Lambda(M)=\bigoplus_{r=0}^{n}\Lambda^{r}(M)$ 
forms a $\integers_2$-graded commutative algebra. 
Needless to say, it is a subalgebra of $\Omega(M)$.

\section{Finding Poisson Bracket}
\label{sec:Finding Poisson Bracket}

Let us find a natural definition of the Poisson bracket of the covariant canonical formalism. 
To avoid complexity, we first restrict ourselves to $n=\text{even}$ case only, 
where $n$ is the dimension of the manifold $M$. 
When $n$ is even, it follows that 
$\varphi^a$ is even, $p_a$ is odd, $\psi^\alpha$ is odd, 
$\pi_\alpha$ is even and $H$ is even. 
Our canonical equations (\ref{ph:canonical_equations}) reduce to 
\begin{align}
  d\varphi^a=-\frac{\pa H}{\pa p_a}, \quad 
  dp_a=-\frac{\pa H}{\pa\varphi^a}, \quad
  d\psi^\alpha=\frac{\pa H}{\pa\pi_\alpha}, \quad 
  d\pi_\alpha=\frac{\pa H}{\pa\psi^\alpha} . 
  \label{ph:canonical_equations_when_n=even}
\end{align}
Let $F\in\Lambda(M)$ be an arbitrary, even or odd, 
differentiable form on the ``phase space''.  
At the moment, we naively consider our phase space 
to be something like a space of which ``coordinates'' are differential forms 
$\varphi^a$, $p_a$, $\psi^\alpha$ and $\pi_\alpha$. 
Its sophisticated definition will be given in Section \ref{sec:Phase Space}. 
We then obtain 
\begin{align*}
  dF 
  &=d\varphi^a\wedge\frac{\pa F}{\pa\varphi^a}
  +dp_a\wedge\frac{\pa F}{\pa p_a}
  +d\psi^\alpha\wedge\frac{\pa F}{\pa\psi^\alpha}
  +d\pi_\alpha\wedge\frac{\pa F}{\pa\pi_\alpha} \\
  &=-\pfa{H}{F}-\fpa{H}{F}+\psa{H}{F}+\spa{H}{F} , 
\end{align*}
where we have used (\ref{ph:canonical_equations_when_n=even}). 
Now, we postulate 
\begin{align}
  dF=-\{ H,F \} 
  \label{ph:eq_of_motion} 
\end{align}
to define our Poisson bracket as 
\begin{align*}
  \{ H,F \}
  =\left(\pfa{H}{F}+\fpa{H}{F}\right)-\left(\psa{H}{F}+\spa{H}{F}\right) .
\end{align*}
Remembering that $H$ is an even-form, we generalise the definition as 
\begin{align}
  \{ E_1,F \}
  =\left(\pfa{E_1}{F}+\fpa{E_1}{F}\right)-\left(\psa{E_1}{F}+\spa{E_1}{F}\right) .
  \label{ph:n=even:EF}
\end{align}
Here, $E_1\in\Lambda(M)$ denotes an arbitrary differentiable even-form. 
In what follows, unless otherwise mentioned, 
$E$, $E_1$, $E_2$, $E_3$ $\in\Lambda(M)$ are all even-forms 
and $O$, $O_1$, $O_2$, $O_3$ $\in\Lambda(M)$ are all odd-forms. 
Putting $F=E_2$, $O_2$ in (\ref{ph:n=even:EF}), we obtain
\begin{align}
  \{ E_1,E_2 \}
  &=\left(\pfa{E_1}{E_2}+\fpa{E_1}{E_2}\right)-\left(\psa{E_1}{E_2}+\spa{E_1}{E_2}\right) , 
  \label{ph:n=even:EE} \\
  \{ E_1,O_2 \}
  &=\left(\pfa{E_1}{O_2}+\fpa{E_1}{O_2}\right)-\left(\psa{E_1}{O_2}+\spa{E_1}{O_2}\right) .
  \label{ph:n=even:EO}
\end{align}
Next, let us put $E_1=O_1\wedge o$ in (\ref{ph:n=even:EE}) and (\ref{ph:n=even:EO}), 
in which $o$ is a constant odd-form, an invariant odd-form under variations of 
canonical variables. 
Then, we obtain 
\begin{align*}
  \{O_1\wedge o, E_2\}
  &=\left[
  \left(\pfa{O_1}{E_2}-\fpa{O_1}{E_2}\right)
  +\left(\psa{O_1}{E_2}-\spa{O_1}{E_2}\right)
  \right]\wedge o , 
  \\
  \{O_1\wedge o, O_2\}
  &=-\left[
  \left(\pfa{O_1}{O_2}-\fpa{O_1}{O_2}\right)
  +\left(\psa{O_1}{O_2}-\spa{O_1}{O_2}\right)
  \right]\wedge o . 
\end{align*}
These results suggest us for putting additional definitions as 
\begin{align}
  &\{O_1, E_2\} 
  = 
  t\left(\pfa{O_1}{E_2}-\fpa{O_1}{E_2}\right)
  +t\left(\psa{O_1}{E_2}-\spa{O_1}{E_2}\right) ,
  \label{ph:n=even:OE} \\
  &\{O_1, O_2\} 
  = 
  t'\left(\pfa{O_1}{O_2}-\fpa{O_1}{O_2}\right)
  +t'\left(\psa{O_1}{O_2}-\spa{O_1}{O_2}\right) . 
  \label{ph:n=even:OO}
\end{align}
Here, $t$ and $t'$ are yet unknown real constants, which we assume to be $1$ or $-1$. 
Now notice that the Poisson bracket we have just defined in 
(\ref{ph:n=even:EE})(\ref{ph:n=even:EO})(\ref{ph:n=even:OE})(\ref{ph:n=even:OO}) 
is an \textit{odd-bracket}; 
that is, $\{ E_1,E_2 \}$ is odd, $\{ E_1,O_2 \}$ is even, 
$\{O_1,E_2 \}$ is even, and $\{ O_1,O_2 \}$ is odd. 
We can fix $t$ and $t'$ postulating ``Leibniz rule'' 
in the form of 
$\{A, \{B, C\} \}=\{ \{A, B\}, C \} \pm \{ B , \{ A, C \}\}$. 
Our bracket, for example, satisfies 
\begin{align*}
  &\{E, \{E_1, O_2\} \}=-t'\{ \{E, E_1\}, O_2 \} -\{ E_1 , \{ E, O_2 \}\}   , \\
  &\{E, \{O_1, E_2\} \}=-t\{ \{E, O_1\}, E_2 \} +tt'\{ O_1 , \{ E, E_2 \}\} .
\end{align*}
Hence, we get $t=t'=-1$.  Then, our bracket has been totally fixed. 

Now, let us turn our attention to $n=\text{odd}$ case. 
We here outline points only 
since we can consider it along the same lines in the above. 
A difference occurs after equation (\ref{ph:eq_of_motion}) 
because $H$ is an odd-form when $n$ is odd. 
That is, when $n$ is odd, the bracket we obtain from (\ref{ph:eq_of_motion}) 
is $\{ O_1, F \}$, instead of $\{ E_1, F \}$ in (\ref{ph:n=even:EF}). 
Other consideration goes well. 
After lengthy calculation, we finally get a unique Poisson bracket for $n=\text{odd}$ case, too. 
Again, notice that when $n$ is odd, our bracket becomes an \textit{even-bracket};  
that is, $\{ E_1,E_2 \}$ is even, $\{ E_1,O_2 \}$ is odd, 
$\{O_1,E_2 \}$ is odd, and $\{ O_1,O_2 \}$ is even. 
We will summarise 
important formulae of our Poisson bracket in the next section. 

Before closing this section, 
note that our Poisson bracket is defined on the space of differential forms. 
Research on the generalisation of the traditional Poisson bracket to a graded Lie bracket 
on the space of differential forms has been a long-standing topic of mathematics 
\cite{Abraham,Koszul,Michor,Cantrijin,Beltran-Monterde,Grabowski,Monterde-Vallejo}. 
Our bracket provides a concrete example of such a generalised bracket 
with physical motivation.

\section{List of Poisson Bracket Formulae}
\label{sec:List of Poisson Bracket Formulae}

In this section, we list, without verbose description, important formulae of our Poisson bracket 
heuristically found in the previous section. 
One can prove all the formulae listed here 
by straightforward calculation using definitions, 
though it is tedious. 
We will give an alternative elegant proof in Section \ref{sec:Symplectic Structure} in this paper. 
Although we have omitted from the following list, the formula (\ref{ph:eq_of_motion}) holds, 
of course, by construction; 
it is the field equation of a differentiable form $F\in\Lambda(M)$ 
with our Poisson bracket and the Hamiltonian form $H$.

\subsection{Poisson Bracket in Even Dimensions}

\noindent
\textit{Definition}
\begin{align}
\begin{split}
  \underset{\text{odd}}{\{ E_1,E_2 \}}
  &=\left(\pfa{E_1}{E_2}+\fpa{E_1}{E_2}\right)-\left(\psa{E_1}{E_2}+\spa{E_1}{E_2}\right) 
  \\
  \underset{\text{even}}{\{ E_1,O_2 \}}
  &=\left(\pfa{E_1}{O_2}+\fpa{E_1}{O_2}\right)-\left(\psa{E_1}{O_2}+\spa{E_1}{O_2}\right) 
  \\
  \underset{\text{even}}{\{O_1, E_2\}}
  &= 
  -\left(\pfa{O_1}{E_2}-\fpa{O_1}{E_2}\right)
  -\left(\psa{O_1}{E_2}-\spa{O_1}{E_2}\right)
  \\
  \underset{\text{odd}}{\{O_1, O_2\}}
  &= 
  -\left(\pfa{O_1}{O_2}-\fpa{O_1}{O_2}\right)
  -\left(\psa{O_1}{O_2}-\spa{O_1}{O_2}\right)
\end{split}\label{ph:n=evenPB}
\end{align}
\textit{Fundamental Brackets}
\begin{align}
  \{\varphi^a, p_b \}=\delta^a_b, \quad 
  \{p_b, \varphi^a \}=-\delta^a_b, \quad 
  \{\psi^\alpha, \pi_\beta \}=\delta^\alpha_\beta, \quad 
  \{\pi_\beta , \psi^\alpha \}=-\delta^\alpha_\beta, \quad 
  \text{others}=0
  \label{ph:n=evenFPB}
\end{align}
\textit{Symmetry}
\begin{align}
  \{E_2, E_1\}=\{E_1, E_2\}, \quad 
  \{O, E\}=-\{E, O\}, \quad 
  \{O_2, O_1\}=-\{O_1, O_2\}
  \label{ph:n=evenSymmetry}
\end{align}
\textit{Leibniz Rule I}
\begin{align}
\begin{split}
  &\{E, E_1\wedge E_2\}=\{E, E_1 \}\wedge E_2 +E_1\wedge \{E, E_2 \} \\
  &\{E, E_1\wedge O_2\}=\{E, E_1 \}\wedge O_2 +E_1\wedge \{E, O_2 \} \\
  &\{E, O_1\wedge E_2\}=\{E, O_1 \}\wedge E_2 -O_1\wedge \{E, E_2 \} \\
  &\{E, O_1\wedge O_2\}=\{E, O_1 \}\wedge O_2 -O_1\wedge \{E, O_2 \} \\
  &\{O, E_1\wedge E_2\}=\{O, E_1 \}\wedge E_2 +E_1\wedge \{O, E_2 \} \\
  &\{O, E_1\wedge O_2\}=\{O, E_1 \}\wedge O_2 +E_1\wedge \{O, O_2 \} \\
  &\{O, O_1\wedge E_2\}=\{O, O_1 \}\wedge E_2 +O_1\wedge \{O, E_2 \} \\
  &\{O, O_1\wedge O_2\}=\{O, O_1 \}\wedge O_2 +O_1\wedge \{O, O_2 \} \\
\end{split}\label{ph:n=evenLeibnizI}
\end{align}
\textit{Jacobi identity}
\begin{align}
\begin{split}
  &\{E_1, \{E_2, E_3 \}\} + \{E_2, \{E_3, E_1 \}\} + \{E_3, \{E_1, E_2 \}\} =0 \\
  &\{E_1, \{E_2, O \}\} -\{E_2, \{O, E_1 \}\} + \{O, \{E_1, E_2 \}\} =0 \\
  & \{E, \{O_1, O_2 \}\} + \{O_1, \{O_2, E \}\} + \{O_2, \{E, O_1 \}\} =0 \\
  &\{O_1, \{O_2, O_3 \}\} + \{O_2, \{O_3, O_1 \}\} + \{O_3, \{O_1, O_2 \}\} =0 \\
\end{split}\label{ph:n=evenJacobi}
\end{align}
\textit{Leibniz Rule I\!I (Another Form of Jacobi Identity)}
\begin{align}
\begin{split}
&\{E, \{E_1, E_2 \}\} = \{\{E, E_1 \}, E_2 \} - \{E_1, \{ E, E_2 \}\} \\
&\{E, \{E_1, O_2 \}\} = \{\{E, E_1 \}, O_2 \} - \{E_1, \{ E, O_2 \}\} \\
&\{E, \{O_1, E_2 \}\} = \{\{E, O_1 \}, E_2 \} + \{O_1, \{ E, E_2 \}\} \\
&\{E, \{O_1, O_2 \}\} = \{\{E, O_1 \}, O_2 \} + \{O_1, \{ E, O_2 \}\} \\
&\{O, \{E_1, E_2 \}\} = \{\{O, E_1 \}, E_2 \} + \{E_1, \{ O, E_2 \}\} \\
&\{O, \{E_1, O_2 \}\} = \{\{O, E_1 \}, O_2 \} + \{E_1, \{ O, O_2 \}\} \\
&\{O, \{O_1, E_2 \}\} = \{\{O, O_1 \}, E_2 \} + \{O_1, \{ O, E_2 \}\} \\
&\{O, \{O_1, O_2 \}\} = \{\{O, O_1 \}, O_2 \} + \{O_1, \{ O, O_2 \}\} 
\end{split}\label{ph:n=evenLeibnizII}
\end{align}

\subsection{Poisson Bracket in Odd Dimensions}

\noindent
\textit{Definition}
\begin{align}
\begin{split}
  \underset{\text{even}}{\{E_1, E_2\}}
  &= 
  -\left(\pfa{E_1}{E_2}-\fpa{E_1}{E_2}\right)
  +\left(\psa{E_1}{E_2}+\spa{E_1}{E_2}\right)
  \\
  \underset{\text{odd}}{\{E_1, O_2\}}
  &= 
  -\left(\pfa{E_1}{O_2}-\fpa{E_1}{O_2}\right)
  +\left(\psa{E_1}{O_2}+\spa{E_1}{O_2}\right)
  \\
  \underset{\text{odd}}{\{ O_1,E_2 \}}
  &=-\left(\pfa{O_1}{E_2}-\fpa{O_1}{E_2}\right)-\left(\psa{O_1}{E_2}+\spa{O_1}{E_2}\right) 
  \\
  \underset{\text{even}}{\{ O_1,O_2 \}}
  &=-\left(\pfa{O_1}{O_2}-\fpa{O_1}{O_2}\right)-\left(\psa{O_1}{O_2}+\spa{O_1}{O_2}\right) 
\end{split}\label{ph:n=oddPB}
\end{align}
\textit{Fundamental Brackets}
\begin{align}
  \{\varphi^a, p_b \}=\delta^a_b, \quad 
  \{p_b, \varphi^a \}=-\delta^a_b, \quad 
  \{\psi^\alpha, \pi_\beta \}=-\delta^\alpha_\beta, \quad 
  \{\pi_\beta , \psi^\alpha \}=-\delta^\alpha_\beta, \quad 
  \text{others}=0
  \label{ph:n=oddFPB}
\end{align}
\textit{Symmetry}
\begin{align}
  \{E_2, E_1\}=-\{E_1, E_2\}, \quad 
  \{O, E\}=-\{E, O\}, \quad 
  \{O_2, O_1\}=\{O_1, O_2\}
  \label{ph:n=oddSymmetry}
\end{align}
\textit{Leibniz Rule I}
\begin{align}
\begin{split}
  &\{E, E_1\wedge E_2\}=\{E, E_1 \}\wedge E_2 +E_1\wedge \{E, E_2 \} \\
  &\{E, E_1\wedge O_2\}=\{E, E_1 \}\wedge O_2 +E_1\wedge \{E, O_2 \} \\
  &\{E, O_1\wedge E_2\}=\{E, O_1 \}\wedge E_2 +O_1\wedge \{E, E_2 \} \\
  &\{E, O_1\wedge O_2\}=\{E, O_1 \}\wedge O_2 +O_1\wedge \{E, O_2 \} \\
  &\{O, E_1\wedge E_2\}=\{O, E_1 \}\wedge E_2 +E_1\wedge \{O, E_2 \} \\
  &\{O, E_1\wedge O_2\}=\{O, E_1 \}\wedge O_2 +E_1\wedge \{O, O_2 \} \\
  &\{O, O_1\wedge E_2\}=\{O, O_1 \}\wedge E_2 -O_1\wedge \{O, E_2 \} \\
  &\{O, O_1\wedge O_2\}=\{O, O_1 \}\wedge O_2 -O_1\wedge \{O, O_2 \} \\
\end{split}\label{ph:n=oddLeibnizI}
\end{align}
\textit{Jacobi identity}
\begin{align}
\begin{split}
  &\{E_1, \{E_2, E_3 \}\} + \{E_2, \{E_3, E_1 \}\} + \{E_3, \{E_1, E_2 \}\} =0 \\
  &\{E_1, \{E_2, O \}\} + \{E_2, \{O, E_1 \}\} + \{O, \{E_1, E_2 \}\} =0 \\
  &\{E, \{O_1, O_2 \}\} + \{O_1, \{O_2, E \}\} - \{O_2, \{E, O_1 \}\} =0 \\
  &\{O_1, \{O_2, O_3 \}\} + \{O_2, \{O_3, O_1 \}\} + \{O_3, \{O_1, O_2 \}\} =0 \\
\end{split}\label{ph:n=oddJacobi}
\end{align}
\textit{Leibniz Rule I\!I (Another Form of Jacobi Identity)}
\begin{align}
\begin{split}
&\{E, \{E_1, E_2 \}\} = \{\{E, E_1 \}, E_2 \} + \{E_1, \{ E, E_2 \}\} \\
&\{E, \{E_1, O_2 \}\} = \{\{E, E_1 \}, O_2 \} + \{E_1, \{ E, O_2 \}\} \\
&\{E, \{O_1, E_2 \}\} = \{\{E, O_1 \}, E_2 \} + \{O_1, \{ E, E_2 \}\} \\
&\{E, \{O_1, O_2 \}\} = \{\{E, O_1 \}, O_2 \} + \{O_1, \{ E, O_2 \}\} \\
&\{O, \{E_1, E_2 \}\} = \{\{O, E_1 \}, E_2 \} + \{E_1, \{ O, E_2 \}\} \\
&\{O, \{E_1, O_2 \}\} = \{\{O, E_1 \}, O_2 \} + \{E_1, \{ O, O_2 \}\} \\
&\{O, \{O_1, E_2 \}\} = \{\{O, O_1 \}, E_2 \} - \{O_1, \{ O, E_2 \}\} \\
&\{O, \{O_1, O_2 \}\} = \{\{O, O_1 \}, O_2 \} - \{O_1, \{ O, O_2 \}\} 
\end{split}\label{ph:n=oddLeibnizII}
\end{align}

\section{Phase Space}
\label{sec:Phase Space}

Let us change our subject slightly. 
In the preceding sections, we have heuristically introduced the Poisson bracket. 
In that context, our bracket seems merely an assembly of ad hoc definitions, 
but nevertheless it has a profound geometrical meaning 
as is shown below. 
In the remaining part of the paper, we turn our attention to mathematical structure 
behind the Poisson bracket.  
After preparing necessary tools in the present and the next sections, 
we will get back to the issue of 
the Poisson bracket in Section \ref{sec:Symplectic Structure}. 

Let us start our consideration from the definition of phase space. 
The phase space of our canonical theory, intuitively, is a space of which ``coordinates'' 
are forms on an $n$-dimensional $C^\infty$-manifold $M$. 
Obviously, we need a sophisticated definition to it to go further. 
Remember that for any open subset $U\subset M$, there exists a short exact sequence 
\begin{align*}
  0 
  \longrightarrow \Lambda'(U) 
  \longrightarrow \Lambda(U)
  \longrightarrow C^\infty(U)
  \longrightarrow 0, 
\end{align*}
where $\Lambda'(U)$ is the ideal 
of nilpotent elements in $\Lambda(U)$, 
a $\integers_2$-graded commutative algebra of differentiable differential forms on $U$. 
Considering this in mind, we identify our phase space with a ringed space $(M,\scrO_M)$ 
with the structure sheaf $\scrO_M: U \mapsto \Lambda(U)$. 
Notice that $(M,\scrO_M)$ is a kind of what is called a supermanifold 
\cite{Kostant,Bartocci,Tuynman,Rogers,Carmeli}. 
In the following, we think of $x^i\in\Lambda(U)$ as any of canonical variables 
$\varphi^a, p_a, \psi^\alpha, \pi_\alpha$. 
We define $\integers_2$-graded degree of $x^i$ as $|x^i|=r$ ($\bmod\, 2$) for $x^i\in\Lambda^r(U)$. 
In our construction, 
the set of ``functions'' on $(M,\scrO_M)$ is identified with $\Lambda(M)$. 
A vector field $\bX$ on $(M,\scrO_M)$ is an $\realnumbers$-linear derivation of $\Lambda(U)$ 
\begin{align*}
  \bX=X^i\wedge\frac{\bpa}{\pa x^i}:~ \Lambda(U)\to\Lambda(U)
\end{align*}
with $X^i\in\Lambda(U)$ such that 
$\bX f=X^i\wedge{\pa f}/{\pa x^i}$ for $f\in\Lambda(U)$. 
The set of all vector fields on $(M,\scrO_M)$ is denoted $\Der\!\Lambda(M)$, 
which is a $\integers_2$-graded left $\Lambda(M)$-module. 
$|\bX|$ implies $\integers_2$-graded degree of $\bX$. 
We define the dual basis $\bd x^i\in \Hom_{\Lambda(U)}(\Der\!\Lambda(U),\Lambda(U))$ of 
$\bpa/\pa x^i \in\Der\!\Lambda(U)$ as 
\begin{align*}
  \left\la \frac{\bpa}{\pa x^j} ; \bd x^i  \right\ra =\frac{\pa x^i}{\pa x^j}=\delta^i_j .
\end{align*}
Differential forms on $(M,\scrO_M)$ can be defined in a straightforward manner 
following generalised techniques of the usual manifold theory  
so as to accommodate $\integers_2$-gradings of the structure sheaf $\scrO_M$. 
The exterior product $\bwedge$ on $(M,\scrO_M)$, 
which we write in baldface to distinguish from the exterior product $\wedge$ on $M$, 
is defined so as to satisfy on $(U,\scrO_U)$ 
\begin{align*}
  &\bd x^i \bwedge \bd x^j = \bd x^i \otimes \bd x^j - (-1)^{x^ix^j} \bd x^j \otimes \bd x^i , \\
  &\bd x^i \bwedge \bd x^j \bwedge \bd x^k 
  = \bd x^i \otimes \bd x^j \otimes \bd x^k 
  + (-1)^{x^i(x^j+x^k)} \bd x^j \otimes \bd x^k \otimes \bd x^i \notag \\ 
  &\qquad\qquad\qquad\quad 
  + (-1)^{(x^i + x^j) x^k} \bd x^k \otimes \bd x^i \otimes \bd x^j 
  - (-1)^{x^i x^j} \bd x^j \otimes \bd x^i \otimes \bd x^k \notag \\ 
  &\qquad\qquad\qquad\quad 
  - (-1)^{x^j x^k} \bd x^i \otimes \bd x^k \otimes \bd x^j 
  - (-1)^{x^j x^k + x^i(x^j + x^k)} \bd x^k \otimes \bd x^j \otimes \bd x^i , 
\end{align*}
and so on. 
Throughout the present paper, 
we abbreviate the degree symbol $|~~ |$ 
in the exponent of $(-1)$ for simplicity. 
As the readers have anticipated, 
we will use boldfaced $\bd$ for the exterior derivative on $(M,\scrO_M)$ 
distinguishing from $d$ on $M$. 
It is important to notice that 
differential forms on $(M,\scrO_M)$, 
which we denote by $\Omega(M,\scrO_M)=\bigoplus_{r=0}^n\Omega^r(M,\scrO_M)$, 
form a $(\integers_2\times\integers_2)$-bigraded right $\Lambda(M)$-module. 
We will say that a differential form $\blambda$ on $(M,\scrO_M)$ has a $(\integers_2\times\integers_2)$-bidegree 
$|\blambda|=\begin{pmatrix}r \\ s \end{pmatrix}$ ($\bmod\, 2$)
if it is an $r$-form $\blambda\in\Omega^r(M,\scrO_M)$ which satisfies  
\begin{align}
 \big| \la \bX_1,\dots,\bX_r;\blambda\ra \big|=\sum_{k=1}^r | \bX_i | + s
\end{align}
for all vector fields $\bX_1,\cdots,\bX_r \in\Der\!\Lambda(M)$. 
Arbitrary forms $\blambda_1,\blambda_2\in\Omega(M,\scrO_M)$ with 
bidegrees 
$|\blambda_1|=\begin{pmatrix}r_1 \\ s_1 \end{pmatrix}$ and  
$|\blambda_2|=\begin{pmatrix}r_2 \\ s_2 \end{pmatrix}$ 
satisfy
\begin{align*} 
  \blambda_1\bwedge\blambda_2=(-1)^{r_1r_2}(-1)^{s_1s_2}\blambda_2\bwedge\blambda_1 . 
\end{align*}
For a differential $r$-form $\blambda\in\Omega^{r}(M,\scrO_M)$ ($r\geq 1$), 
and a vector field $\bX\in\Der\!\Lambda(M)$, 
we define the interior product of $\blambda$ and $\bX$, 
$i_{\bX}\blambda\in\Omega^{r-1}(M,\scrO_M)$, by 
\begin{align*}
  \la \bX_1,\cdots,\bX_{r-1};i_{\bX}\blambda \ra = \la \bX_1,\cdots,\bX_{r-1},\bX;\blambda\ra 
\end{align*}
with $\bX_1,\cdots,\bX_{r-1}\in\Der\!\Lambda(M)$. 
If $\blambda$ has a bidegree $\begin{pmatrix} r \\ s \end{pmatrix}$, 
then $i_\bX\blambda$ 
has a bidegree $\begin{pmatrix} r-1 \\ s+|\bX| \end{pmatrix}$ ($\bmod\, 2$).

\section{Lie Derivatives}
\label{sec:Lie Derivatives}

The Lie bracket of two vector fields $\bX, \bY\in\Der\!\Lambda(M)$ is 
the unique vector field defined as 
\begin{align*}
  [\bX, \bY]=\bX\bY - (-1)^{\bX\bY}\bY\bX  ~\in\Der\!\Lambda(M) . 
\end{align*}
The real vector space structure of $\Der\!\Lambda(M)$
endowed with the Lie bracket $[~ ,~ ]$ on $\Der\!\Lambda(M)$ 
form a $\integers_2$-graded Lie algebra. That is, 
$[~, ~]$ is $\realnumbers$-bilinear 
and, for any $\bX,\bY,\bZ\in\Der\!\Lambda(M)$, 
there exist symmetry 
\begin{align*}
  [\bY, \bX]=-(-1)^{\bX\bY} [\bX, \bY] 
\end{align*}
and Jacobi identity 
\begin{align*}
  (-1)^{\bX\bZ}[\bX, [\bY,\bZ]] 
  + (-1)^{\bY\bX}[\bY, [\bZ,\bX]] 
  + (-1)^{\bZ\bY}[\bZ, [\bX,\bY]] =0 .
\end{align*} 
For a given vector field $\bX\in\Der\!\Lambda(M)$, we define a map 
$L_\bX: \Der\!\Lambda(M)\to\Der\!\Lambda(M)$, 
the Lie derivative on $\Der\!\Lambda(M)$, as 
\begin{align}
  L_{\bX}\bY=[\bX,\bY] \quad (\bY\in\Der\!\Lambda(M)) .
  \label{math:LieDer_vec}
\end{align}
The following are trivial restatements of the properties just mentioned above.  
$\realnumbers$-bilinearlity 
\begin{align*}
  &L_{\bX}(c_1\bY_1+c_2\bY_2)=c_1L_{\bX}\bY_1+c_2L_{\bX}\bY_2 , \\  
  &L_{c_1\bX+c_2\bX_2}\bY=c_1L_{\bX_1}\bY+c_2L_{\bX_2}\bY , 
\end{align*}
where $c_1,c_2\in\realnumbers$, 
symmetry
\begin{align*}
  L_{\bY}\bX=-(-1)^{\bX\bY}L_{\bX}\bY , 
\end{align*}
and Leibniz rule (Jacobi identity)
\begin{align*}
  L_{\bX}[\bY,\bZ]=[L_{\bX}\bY,\bZ]+(-1)^{\bX\bY}[\bY,L_{\bX}\bZ] . 
\end{align*}

We can naturally generalise the above mentioned Lie derivative 
$L_\bX$ on $\Der\!\Lambda(M)$ 
to obtain another kind of Lie derivative 
$L_\bX: \Lambda(M)\to\Lambda(M)$, 
the Lie derivative on $\Lambda(M)$, so that 
\begin{align}
  L_\bX (f\wedge\bY)=L_\bX f\wedge \bY + (-1)^{\bX f} f\wedge L_\bX \bY 
\end{align}
for any $f\in\Lambda(M)$ and $\bX,\bY\in\Der\!\Lambda(M)$. 
The resultant definition of $L_\bX: \Lambda(M)\to\Lambda(M)$ is 
\begin{align}
  L_\bX f = \bX f \quad (f\in\Lambda(M)) .
\end{align}
The Lie derivative on $\Lambda(M)$, just defined, is a derivation on the 
$\integers_2$-graded commutative algebra $\Lambda(M)$. 
That is, it is $\realnumbers$-linear and satisfies Leibniz rule 
\begin{align}
  L_\bX(f\wedge g)
  =\left( L_\bX f \right) \wedge g 
  +(-1)^{\bX f} f\wedge \left( L_\bX g \right) 
\end{align}
for $f,g\in\Lambda(M)$.  The following formula is easily verified.  
\begin{align}
  L_{[ \bX, \bY]}f=[L_\bX,L_\bY]f
\end{align}

Similarly, there is a natural extension of 
the Lie derivatives $L_\bX$ on $\Der\!\Lambda(M)$ and on $\Lambda(M)$ 
to $L_\bX: \Omega(M,\scrO_M)\to\Omega(M,\scrO_M)$, 
the Lie derivative on $\Omega(M,\scrO_M)$. 
That is, for $\blambda\in\Omega^r(M,\scrO_M)$,  
we define $L_\bX \blambda\in\Omega^r(M,\scrO_M)$ so that 
\begin{align*}
  &L_\bX \la \bY_1,\bY_2,\cdots,\bY_r ; \blambda \ra 
  = \la L_\bX \bY_1,\bY_2,\cdots,\bY_r ; \blambda \ra 
  +(-1)^{\bX\bY_1} \la\bY_1, L_\bX \bY_2,\cdots,\bY_r ; \blambda \ra 
  +\cdots \notag \\ 
  &
  +(-1)^{\bX(\bY_1+\bY_2+\cdots+\bY_{r-1})} \la\bY_1, \bY_2,\cdots,L_\bX \bY_r ; \blambda \ra 
  + (-1)^{\bX(\bY_1+\bY_2+\cdots+\bY_r)} \la \bY_1,\bY_1,\cdots,\bY_r ; L_\bX \blambda \ra
\end{align*}
for any $\bY_1,\cdots,\bY_r\in\Der\!\Lambda(M)$.  
The resultant definition of $L_\bX: \Omega(M,\scrO_M)\to\Omega(M,\scrO_M)$ turns out 
to be 
\begin{align}
  L_{\bX}\blambda=(-1)^r(i_\bX \bd -\bd i_\bX)\blambda 
  \quad (\blambda\in\Omega^r(M,\scrO_M))
  \label{math:L_X(lambda)}
\end{align}
with $r\geq 1$. 
Note for the unusual sign in this formula. 
If $\blambda$ has a bidegree $\begin{pmatrix} r \\ s \end{pmatrix}$, 
its Lie derivative $L_\bX\blambda$ has a bidegree $\begin{pmatrix} r \\ s+|\bX| \end{pmatrix}$ ($\bmod\, 2$). 
Lie derivative on $\Omega(M,\scrO_M)$ satisfies 
\begin{align}
  i\left( L_{\bX}\bY \right)\blambda 
  =\left( L_{\bX}i_{\bY}-(-1)^{\bX\bY}i_{\bY}L_{\bX} \right) \blambda 
  \label{math:iL_X(Y)}
\end{align}
where $i(\bZ)\blambda$ implies $i_\bZ\blambda$.

\section{Symplectic Structure}
\label{sec:Symplectic Structure}

We now define our \textit{symplectic form} $\bomega\in\Omega^2(M,\scrO_M)$ so that 
\begin{align}
  \bomega 
  &= 
  -\bd \varphi^a \bwedge \bd p_a 
  -\bd \psi^\alpha \bwedge \bd \pi_\alpha 
\end{align}
on $(U,\scrO_U)$. 
$\bomega$ is a closed, non-degenerate 2-form 
with a $(\integers_2\times\integers_2)$-bidegree 
$|\bomega|= \begin{pmatrix} 2 \\ n+1 \end{pmatrix}$ ($\bmod\, 2$). 
Note that when $n$ is even, $\bomega$ becomes what is called an \textit{odd symplectic form}, 
and when $n$ is odd, an \textit{even symplectic form} \cite{Rogers}.  

The symplectic structure $\bomega$ leads to natural identification between 
a vector field on $(M,\scrO_M)$ and a 1-form on $(M,\scrO_M)$. 
Indeed, if we define a map $\flat: \Der\!\Lambda(M)\to\Omega^1(M,\scrO_M)$, $\bX\mapsto i_{\bX}\bomega$, 
namely 
\begin{align}
  \flat(\bX)=i_\bX \bomega , 
\end{align}
then the map $\flat$ gives a linear isomorphism between $\Der\!\Lambda(M)$ and $\Omega^1(M,\scrO_M)$. 
We denote by $\sharp$ the inverse map of $\flat$. 
\begin{Def}
The Hamiltonian vector field $\bX_f\in\Der\!\Lambda(M)$ generated by 
an arbitrary differentiable differential form $f\in\Lambda(M)$ 
is the unique vector field on $(M,\scrO_M)$ determined by $\bX_f=\sharp\, \bd f$. 
\end{Def}
Note that the Hamiltonian vector field $\bX_f$ has a $\integers_2$-graded degree $|\bX_f|=|f|+n+1$ ($\bmod\, 2$). 
By definition, it satisfies 
\begin{align}
  \bd f= i_{\bX_f}\bomega .
  \label{math:d(f)=i(X_f)(omega)}
\end{align}
With canonical coordinates 
$\varphi^a$, $p_a$, $\psi^\alpha$, $\pi_\alpha$ $\in\Lambda(U)$, 
$\bX_f$ can be written as 
\begin{align}
  \bX_f
  &=(-1)^{(n+1)(f+1)}\frac{\pa f}{\pa p_a}\wedge\frac{\bpa}{\pa\varphi^a}
  -\frac{\pa f}{\pa\varphi^a}\wedge\frac{\bpa}{\pa p_a} \notag \\
  &\qquad 
  +(-1)^{n(f+1)}\frac{\pa f}{\pa\pi_\alpha}\wedge\frac{\bpa}{\pa\psi^\alpha}
  +(-1)^{f+n}\frac{\pa f}{\pa\psi^\alpha}\wedge\frac{\bpa}{\pa\pi_\alpha} .
  \label{math:HamiltonianVF}
\end{align}
\begin{Def}
The Poisson bracket of two differentiable differential forms $f,g\in\Lambda(M)$ is 
the differentiable differential form $\{f,g\}\in\Lambda(M)$ defined as 
\begin{align}
\begin{split}
  \{f,g \}
  &:=-\la \bX_f,\bX_g;\bomega \ra
  =-\la \bX_f;i_{\bX_g}\omega \ra
  \\ &\qquad 
  =-i_{\bX_f}i_{\bX_g}\bomega 
  =-i_{\bX_f}\bd g
  =-\la \bX_f;\bd g\ra
  =-\bX_f g
  =-L_{\bX_f}g . 
\end{split}\label{math:PB_def}
\end{align}
\end{Def}
The Poisson bracket $\{f, g\}$ has a $\integers_2$-graded degree 
$|\{ f,g \}|=|f|+|g|+n+1$ ($\bmod\, 2$). 
The real vector space structure of $\Lambda(M)$
endowed with the Poisson bracket $\{~ ,~ \}$ 
form a $\integers_2$-graded Lie algebra, 
as it will be obvious from formulae in what follows. 
The definition $\{f,g \}=-\bX_f g$ of the Poisson bracket 
together with (\ref{math:HamiltonianVF}) directly yields 
\begin{align}
\begin{split}
  \{f,g\}
  &=-(-1)^{(n+1)(f+1)}\frac{\pa f}{\pa p_a}\wedge\frac{\pa g}{\pa\varphi^a}
  +\frac{\pa f}{\pa\varphi^a}\wedge\frac{\pa g}{\pa p_a} \\
  &\qquad 
  -(-1)^{n(f+1)}\frac{\pa f}{\pa\pi_\alpha}\wedge\frac{\pa g}{\pa\psi^\alpha}
  -(-1)^{f+n}\frac{\pa f}{\pa\psi^\alpha}\wedge\frac{\pa g}{\pa\pi_\alpha} . 
\end{split}\label{math:PB} 
\end{align}
This result reproduces the heuristic definitions of 
our Poisson bracket, all of (\ref{ph:n=evenPB}) and (\ref{ph:n=oddPB}), completely. 
\begin{Prop} 
The assignment $f\mapsto\bX_f$ yields a $\integers_2$-graded Lie algebra homomorphism from 
$\Lambda(M)$ to $\Der\!\Lambda(M)$, that is, 
\begin{align}
  \bX_{ \{ f,g \} } = -\left[ \bX_f, \bX_g \right] .
  \label{math:X_(f,g)=-[Xf,Xg]}
\end{align}
\end{Prop}
\begin{proof}
Noticing 
\begin{align}
  &L_{\bX_f}\bomega
  =\left( i_{\bX_f}\bd - \bd i_{\bX_f} \right) \bomega 
  =-\bd i_{\bX_f}\bomega
  =-\bd(\bd f)
  =0 ,
  \label{math:L_bX(omega)=0} \\
  &L_{\bX}\bd g 
  =-\left( i_{\bX}\bd - \bd i_{\bX} \right) \bd g
  =\bd i_\bX \bd g
  =\bd \la \bX; \bd g \ra 
  =\bd \bX g 
  =\bd L_\bX g ,
  \label{math:(L_X)dg=d(L_X)g}
\end{align}
we obtain 
\begin{align*}
  &i\left( -[\bX_f,\bX_g] \right) \bomega 
  =-i\left( [\bX_f,\bX_g] \right) \bomega 
  =-i\left( L_{\bX_f} \bX_g \right) \bomega 
  \notag \\ 
  &\quad 
  =-\left( L_{\bX_f} i_{\bX_g}  - (-1)^{\bX_f\bX_g} i_{\bX_g} L_{\bX_f} \right)\bomega  
  =- L_{\bX_f} i_{\bX_g} \bomega 
  \notag \\ 
  &\quad\quad 
  =-L_{\bX_f}\bd g 
  =-\bd L_{\bX_f}g
  =\bd \{ f,g\}
  =i\left( \bX_{\{ f, g \}} \right) \bomega  .
\end{align*}
This completes the proof. 
\end{proof}
\vspace{2ex}

Now, it is an easy task to give concise proofs to many formulae 
listed in Section \ref{sec:List of Poisson Bracket Formulae}. 

\noindent
\textit{Symmetry} (\ref{ph:n=evenSymmetry})(\ref{ph:n=oddSymmetry}): 
\begin{align}
  \{g,f\}
  &=-\la \bX_g,\bX_f;\bomega\ra \notag \\
  &=-(-1)^{(f+n+1)(g+n+1)}\la \bX_f,\bX_g;\bomega\ra \notag \\
  &=-(-1)^{(f+n+1)(g+n+1)}\{f,g\} .
\end{align}
\textit{Leibniz rule I} (\ref{ph:n=evenLeibnizI})(\ref{ph:n=oddLeibnizI}): 
\begin{align}
  \{f,g\wedge h\}
  &=-L_{\bX_f}(g\wedge h) \notag \\
  &=\left( -L_{\bX_f}g \right) \wedge h 
  + (-1)^{(f+n+1)g} g \wedge \left( -L_{\bX_f}g \right) \notag \\ 
  &=\{ f,g \} \wedge h + (-1)^{(f+n+1)g} g \wedge \{ f,h \} .
\end{align}
\textit{Leibniz rule I\!I} (\ref{ph:n=evenLeibnizII})(\ref{ph:n=oddLeibnizII}): 
\begin{align}
  &\{ \{f,g \}, h\}
  =-\bX_{ \{ f,g \} }h
  =[\bX_f , \bX_g ]h
  =\bX_f \bX_g h - (-1)^{\bX_f \bX_g} \bX_g \bX_f h \notag \\
  &\quad 
  =-\bX_f \{g,h\} +(-1)^{\bX_f \bX_g} \bX_g \{f,h\} 
  =\{f, \{ g,h \} \} - (-1)^{\bX_f \bX_g} \{g, \{ f,h \} \} , 
\end{align}
from which we obtain Leibniz rule I\!I, that is, 
\begin{align}
  \{ f, \{ g, h \} \}
  =\{\{ f, g\}, h\} + (-1)^{(f+n+1)(g+n+1)} \{g, \{ f, h\}\} . 
\end{align}
Putting this formula into a symmetric form, we get 
\begin{align}
  &(-1)^{(f+n+1)(h+n+1)}\{f, \{ g, h \}\} + (-1)^{(g+n+1)(f+n+1)} \{g, \{ h, f \}\} \notag \\ 
  &\qquad\qquad\qquad\qquad\qquad\qquad\qquad 
  + (-1)^{(h+n+1)(g+n+1)} \{h, \{ f, g \}\} =0 , 
\end{align}
which, of course, is \textit{Jacobi identity} (\ref{ph:n=evenJacobi})(\ref{ph:n=oddJacobi}).

\section*{Acknowledgements}
The author would like to thank Prof 
Jos\'e Miguel Figueroa-O'Farrill 
for useful comments as well as 
for his warm hospitality at the University of Edinburgh. 
He is also grateful to the National Institute of Technology for the financial support.


\end{document}